\documentclass[11pt]{article}
\usepackage{url}
\usepackage{amssymb,amscd,amstext,amsmath, amsthm}
\usepackage{latexsym}
\usepackage{multirow}

\renewenvironment{abstract}{%
\begin{center}\begin{minipage}{0.9\textwidth}\begin{small}
\textbf{Abstract.}}
{\end{small}\par\noindent\end{minipage}\end{center}}

\newtheorem{theorem}{Theorem}

\newtheorem{corollary}{Corollary}

\newtheorem{proposition}{Proposition}

\newtheorem{conjecture}{Conjecture}

\newtheorem{remark}{Remark}

\newtheorem*{proofbody}{Proof}
\renewenvironment{proof}{\begin{proofbody}\normalfont}{\hfill \qed \end{proofbody}}

\title{\bf A note on the properties of associated Boolean functions of quadratic APN functions
}

\author{A.~Gorodilova
    \\
  \\
  {\small Sobolev Institute of Mathematics, Novosibirsk, Russia} \\
  {\small Novosibirsk State University, Novosibirsk, Russia}\\
    \\
    {\small E-mail: {\tt gorodilova@math.nsc.ru}}
}

\date{}

\begin{document}

\pagenumbering{arabic}

\maketitle
\begin{abstract}
Let $F$ be a quadratic APN function of $n$ variables. The associated Boolean function $\gamma_F$ in $2n$ variables ($\gamma_F(a,b)=1$ if $a\neq{\bf 0}$ and equation $F(x)+F(x+a)=b$ has solutions) has the form $\gamma_F(a,b) = \Phi_F(a) \cdot b + \varphi_F(a) + 1$ for appropriate functions $\Phi_F:\mathbb{F}_2^n\to \mathbb{F}_2^n$ and $\varphi_F:\mathbb{F}_2^n\to \mathbb{F}_2$. We summarize the known results and prove new ones regarding properties of $\Phi_F$ and $\varphi_F$. For instance, we prove that degree of $\Phi_F$ is either $n$ or less or equal to $n-2$. Based on computation experiments, we formulate a conjecture that degree of any component function of $\Phi_F$ is $n-2$. We show that this conjecture is based on two other conjectures of independent interest.
\vspace{0.2cm}

\end{abstract}

\bigskip

\section{Introduction}

Let $\mathbb{F}_2^n$ be the $n$-dimensional vector space over $\mathbb{F}_2$. Let ${\bf 0}$ denote the zero vector of $\mathbb{F}_2^n$ and ${\bf 1}$ denote the vector with all 1s. By $+$ we denote the coordinate-wise sum modulo 2 for vectors from~$\mathbb{F}_2^n$. Let $x \cdot y = x_1y_1 + \ldots + x_ny_n$ denote the {\it inner product} of vectors $x = (x_1,\ldots,x_n),y = (y_1,\ldots,y_n)\in\mathbb{F}_2^n$; $x\preceq y$ if $x_i\leqslant y_i$  for all $i=1,\ldots,n$; and ${\rm wt}(x) = \sum_{i=1}^nx_i$ denote the {\it Hamming weight} of $x\in\mathbb{F}_2^n$. A set $M\subseteq\mathbb{F}_2^n$ forms a {\it linear subspace} if $x+y\in M$ for any $x,y\in M$; the {\it dimension} of $M$, ${\rm dim}(M)$, is the maximal number of linearly independent over $\mathbb{F}_2$ vectors from $M$.
We consider {\it  vectorial Boolean functions} $F:\mathbb{F}_2^n\to \mathbb{F}_2^m$, $F = (f_1, \ldots, f_m)$, where $f_i:\mathbb{F}_2^n\to \mathbb{F}_2$, $i=1,\ldots,m$, is a {\it coordinate} function of $F$.
The {\it algebraic normal form} (ANF) of $F$ is the following unique representation:
$F(x) = \sum_{I\in \mathcal{P}(N)} a_{I} \big(\prod_{i\in I} x_i \big),$
where $\mathcal{P}(N)$ is the power set of $N=\{1,\ldots,n\}$ and $a_I\in\mathbb{F}_2^m$.
The {\it algebraic degree} of $F$ is degree of its ANF: ${\deg}(F) = \max\{|I|:\ a_I\neq{\bf 0},\ I\in\mathcal{P}(N)\}$. Function of algebraic degree at most 1 are called {\it affine} (they are {\it linear} in case of $F({\bf 0}) = {\bf 0}$). Functions of algebraic degree 2 are called {\it quadratic}.
The {\it Walsh transform} $W_f:\mathbb{F}_2^n\to\mathbb{Z}$ of a Boolean function $f:\mathbb{F}_2^n\to\mathbb{F}_2$ is defined as $W_f(u) = \sum_{x\in\mathbb{F}_2^n} (-1)^{f(x)+u\cdot x}$. For $F$ the {\it Walsh spectrum} consists of all {\it Walsh coefficients} $W_{F_v}(u)$, $u\in\mathbb{F}_2^n$, $v\in\mathbb{F}_2^m$, $v\neq{\bf 0}$, where $F_v = v\cdot F$ is a {\it component} Boolean function of $F$.

A function $F$ from $\mathbb{F}_2^n$ to itself is called {\it almost perfect nonlinear} (APN) (according to K.~Nyberg\cite{94-nyberg}) if for any $a,b\in\mathbb{F}_2^n$, $a\neq{\bf 0}$, equation $F(x) + F(x+a) = b$ has at most 2 solutions. APN functions are of special interest for using as S-boxes in block ciphers due to their optimal differential characteristics. Despite the fact that APN functions are intensively studied (see, for example, the book~\cite{RefBud-14} of L.~Budaghyan, surveys \cite{RefP-16} of A.~Pott, \cite{RefGl-16} of M.\,M. Glukhov, \cite{RefT-09} of M.\,E.~Tuzhilin), there are a lot of open problems on finding new constructions, classifications, etc.

In \cite{98-ccz} C.~Carlet, P.~Charpin and V.~Zinoviev introduced the {\it associated Boolean function} $\gamma_F:\mathbb{F}_2^{2n}\to\mathbb{F}_2$ for a given vectorial Boolean function $F$ from $\mathbb{F}_2^n$ to itself; $\gamma_F(a,b)=1$ if and only if $a\neq{\bf 0}$ and equation $F(x)+F(x+a)=b$ has solutions.

Two functions are called {\it differentially equivalent} \cite{19-gorodilova} (or $\gamma$-equivalent according to K.~Boura et al \cite{19-boura}) if their associated Boolean functions coincide. The problem of describing the differential equivalence class of an APN function remains open even for quadratic case. That is why we are interested in obtaining some properties of $\gamma_F$. We will focus on quadratic APN functions.

Let $F$ be a quadratic APN function. Then the set $B_a(F) = \{F(x) + F(x+a)\ |\ x\in\mathbb{F}_2^n\}$ is a linear subspace of dimension $n-1$ or its complement for a nonzero $a\in\mathbb{F}_2^n$. Using this fact, $\gamma_F$ can be uniquely represented in the form $$\gamma_F(a,b) = \Phi_F(a) \cdot b + \varphi_F(a) + 1,$$
where $\Phi_F:\mathbb{F}_2^n\to \mathbb{F}_2^n$, $\varphi_F:\mathbb{F}_2^n\to \mathbb{F}_2$ are defined from
$
B_a(F) = \{y\in\mathbb{F}_2^n:\ \Phi_F(a)\cdot y = \varphi_F(a)\}
$
for all $a\neq {\bf 0}$; and $\Phi_F({\bf 0}) = {\bf 0}$, $\varphi_F({\bf 0})=1$. Note that $B_a(F)$ is a linear subspace if and only if $\varphi_F(a) = 0$. It is easy to see that
$(F(x) + F(x+a) + F(a) + F({\bf 0})) \cdot \Phi_F(a) = 0$ for all $x\in\mathbb{F}_2^n$ by definition.

 In this note we study the properties of functions $\Phi_F$ and $\varphi_F$.

\bigskip
\section{Properties of $\varphi_F$ and $\Phi_F$}

In this section we summarize known results and present new ones about properties of $\Phi_F$ and $\varphi_F$. As it usually happens the cases of even and odd number of variables are different.

\bigskip
\subsection{The image set of $\Phi_F$}

According to \cite{16-Gor-spectr}, let us denote $A_v^F = \{a\in\mathbb{F}_2^n\ |\ \Phi_F(a) = v\}$.

\begin{theorem}[\cite{98-ccz, 16-Gor-spectr}]\label{image}
Let $F$ be a quadratic APN function of $n$ variables.
\begin{enumerate}
\item If $n$ is odd, then $\Phi_F$ is a permutation.
\item If $n$ is even, then the preimage $\Phi_F$ of any nonzero vector is a linear subspace of even dimension together with the zero vector.
\end{enumerate}
\end{theorem}

Note that theorem \ref{image}\,(1) means also that $\gamma_F$ is a bent function of Maiorana–McFarland type (readers may find details regarding bent functions in \cite{15-tokareva}).

\begin{corollary}
Let $F$ be a quadratic APN function. Then $\Phi_F$ takes an odd number of distinct nonzero values.
\end{corollary}
\begin{proof}
By definition of $\Phi_F$, we have $\Phi_F({\bf 0}) = {\bf 0}$.

If $n$ is odd, then $\Phi_F$ is a permutation~\cite{98-ccz}. Hence, the proposition holds.

Let $n$ be even. It is known~\cite{16-Gor-spectr} that the preimage set $A_v^F = \{x\in\mathbb{F}_2^n\ |\ \Phi_F(x) = v\}$ for any nonzero $v\in\mathbb{F}_2^n$ represents a linear subspace of even dimension together with the zero vector. Let $\Phi_F\in\{{\bf 0}, v_1, \ldots, v_m\}$, where $v_i$, $i=1,\ldots,m$, are pairwise distinct nonzero vectors. We need to prove that $m$ is odd.

We have that
$$
2^n - 1 = |A_{v_1}^F| + \ldots + |A_{v_m}^F| = 2^{\lambda_1} - 1 + \ldots + 2^{\lambda_m} - 1 = 2^{\lambda_1} + \ldots + 2^{\lambda_m} - m,
$$
where $\lambda_i$, $i=1,\ldots,m$, is a nonzero even number. Since $2^n-1$ is odd, then $m$ is also odd.
\end{proof}

\bigskip
\subsection{The degree of $\varphi_F$}
\begin{proposition}\label{deg_phi}
Let $F$ be a quadratic APN function of $n$ variables, $n$ is even. Then $\deg(\varphi_F) = n$, or, equivalently, ${\rm wt}(\varphi_F)$ is odd.
\end{proposition}
\begin{proof}
It is known \cite{16-Gor-spectr} that $A_v^F\cup\{{\bf 0}\}$ is a linear subspace of even dimension if $n$ is even for any nonzero $v\in\mathbb{F}_2^n$. Also \cite{16-Gor-spectr}, there exists $c_v\in\mathbb{F}_2^n$ such that $\varphi_F|_{A_v^F} = c_v\cdot x|_{A_v^F}$.
Hence, ${\rm wt}(\varphi_F|_{A_v^F})$ is an even number equal to 0 or $2^{{\rm dim}(A_v^F\cup\{{\bf 0}\}) - 1}$ for any nonzero $v$ and $\varphi_F({\bf 0}) = 1$ by definition.
Thus, ${\rm wt}(\varphi_F)$ is odd. It is widely known that ${\rm wt}(f)$ is odd if and only if $\deg(f) = n$ for any Boolean function of $n$ variables.
\end{proof}

The case of odd $n$ remains open. Based on our computational experiments for all known quadratic APN functions of not more than 11 variables, we can formulate the following
\begin{conjecture}
Let $F$ be a quadratic APN function of $n$ variables, $n$ is odd. Then $\deg(\varphi_F) < n$, or, equivalently, ${\rm wt}(\varphi_F)$ is even.
\end{conjecture}

\bigskip
\subsection{The degree of $\Phi_F$}

\begin{theorem}[\cite{19-gorodilova}]\label{degree-odd}
Let $F$ be a quadratic APN function of $n$ variables, $n\geq3$, $n$ is odd. Then $\deg(\Phi_F)\leq n-2$.
\end{theorem}

The following theorem contains a similar bound for even $n$.

\begin{theorem}\label{degree-even}
Let $F$ be a quadratic APN function in $n$ variables, $n\geq4$, $n$ is even. Then each coordinate function of $\Phi_F$ is represented as $(\Phi_F)_i(x) = f_i(x) + \lambda_i\big(x_2\ldots x_n + x_1x_3\ldots x_n + \ldots + x_1x_2\ldots x_{n-1} +  x_1\ldots x_n\big)$, where $\deg(f_i)\leq n-2$ and $\lambda_i\in\mathbb{F}_2$.
\end{theorem}
\begin{proof}
Let $L:\mathbb{F}_2^n\to\mathbb{F}_2^n$ be a linear function. Then it is easy to see that
$$
\gamma_{F+L}(a,b) = \gamma_F(a, b + L(a)) = (b+L(a))\cdot \Phi_F(a) + \varphi_F(a) + 1 = b\cdot \Phi_F(a) + \varphi_F(a) + L(a) \cdot \Phi_F(a) + 1.
$$
Hence, $\Phi_{F+L} = \Phi_F$ and $\varphi_{F+L} = \varphi_F + L \cdot \Phi_F$.
By proposition~\ref{deg_phi}, $\deg(\varphi_F) = \deg(\varphi_{F+L}) = n$, since $F+L$ is also a quadratic APN function. Thus, $\deg(L\cdot \Phi_F) < n$ for any linear function~$L$.

\underline{Suppose that $\deg(\Phi_F) = n$.} This means that there exists a coordinate function $(\Phi_F)_i$ of degree $n$. Let us represent $$(\Phi_F)_i(x) = f_i(x) + a_1 x_2\ldots x_n + a_2 x_1x_3\ldots x_n + \ldots + a_n x_1x_2\ldots x_{n-1} + x_1\ldots x_n,$$ where $\deg(f_i)\leq n-2$ and $a_1,\ldots, a_n\in\mathbb{F}_2$.

\begin{itemize}
\item If $a_j = 0$, then $\deg(L\cdot \Phi_F) = n$ for $L = (0, \ldots, 0, x_j, 0, \ldots, 0)$, where $x_j$ is the $i$-th coordinate function of $L$. Hence, we get a contradiction.

\item If $a_j=1$ for all $j$, then it is easy to see that we will always have $\deg(L\cdot \Phi_F) < n$ for any linear function $L$.
\end{itemize}

\underline{Suppose that $\deg(\Phi_F) = n-1$.} Similarly,
$$(\Phi_F)_i(x) = f_i(x) + a_1 x_2\ldots x_n + a_2 x_1x_3\ldots x_n + \ldots + a_n x_1x_2\ldots x_{n-1},$$
where at least one coefficient is equal to 1, say $a_j$. Then $\deg(L\cdot \Phi_F) = n$ for $L = (0, \ldots, 0, x_j, 0, \ldots, 0)$, where $x_j$ is the $i$-th coordinate function of $L$. Hence, we get a contradiction.

Thus, $(\Phi_F)_i$ is of degree not more than $n-2$ or all monomials of degree $n-1$ and $n$ are included in the ANF of $(\Phi_F)_i$.
\end{proof}

\begin{remark} For all known quadratic APN functions of not more than 11 variables, we computationally verified that
\begin{itemize}
\item for even $n$, the case $\deg((\Phi_F)_i) = n$ is not realized;
\item any component function of $\Phi_F$ has degree exactly $n-2$.
\end{itemize}
\end{remark}

Based on our computational experiments we can formulate the following
\begin{conjecture}\label{deg}
Let $F$ be a quadratic APN function of $n$ variables, $n\geq3$. Then $\deg(v\cdot\Phi_F) = n-2$ for any nonzero $v\in\mathbb{F}_2^n$.
\end{conjecture}

\bigskip
\section{Does the equality $\deg(\Phi_F) = n-2$ hold?}

In this section we study the following question: ``Is conjecture~\ref{deg} true or not?''.

For example, consider an APN Gold function $F(x) = x^{2^k+1}$, $\gcd(n,k)=1$ (the function is given as a function over the finite field of order $2^n$). Its associated Boolean function is known \cite{98-ccz}: $\gamma_F(a,b) = tr((a^{2^k+1})^{-1}b) + tr(1) + 1$ (here $tr$ is the absolute trace function in the finite field of order $2^n$). So, we have $\Phi_F(a) = (a^{2^k+1})^{-1}$, $\Phi_F(0)=0$, and as it is easy to see $\deg(\Phi_F) = n-2$ (since it is well-known that the degree of a function $F(x) = x^d$ is equal to the 2-weight of the integer $d$ modulo $2^n$).

We wonder whether conjecture~\ref{deg} is true or not for arbitrary $n$. Let us focus on the case of odd $n$ since in this case we have the bound of theorem~\ref{degree-odd}. For even case, the consideration could be rather similar but with assumption that $\deg(\Phi_F)$ is not equal to $n$, that is only a conjecture up to now.

{\bf Step 1.} Let $F$ be a quadratic APN function of $n$ variables, $n$ is odd, $n\geq5$; $v$ be a nonzero vector from $\mathbb{F}_2^n$. We need to prove that $\deg(v\cdot\Phi_F) = n-2$ for any nonzero $v\in\mathbb{F}_2^n$.

We use the following widely known equality for counting the ANF coefficients of a Boolean function $f$ of $n$ variables:
\begin{equation}\label{ANF}
g_{f}(a) = \Big( 2^{{\rm wt}(a) - 1} - 2^{ {\rm wt}(a)-n-1 } \sum_{b\preceq(a+ {\bf 1})} W_f(b)\Big)\ {\rm mod}\ 2.
\end{equation}
 We need to show now that there exists a vector $a^v$ with ${\rm wt}(a^v)=n-2$  such that $g_{v \cdot \Phi_F}(a^v) = 1$. Equivalently, that there exist coordinates $i, j$, $1\leqslant i\neq j\leqslant n$, such that $$\sum_{b\preceq(a^v+ {\bf 1})} W_{v \cdot \Phi_F}(b) = W_{v \cdot \Phi_F}({\bf 0}) + W_{v \cdot \Phi_F}(e^i) + W_{v \cdot \Phi_F}(e^j) + W_{v \cdot \Phi_F}(e^i + e^j)$$ is not divided by 16 according to (\ref{ANF}). Here $e^i$ is the vector with 1 in the $i$-th coordinate and 0s in other coordinates. Let us introduce the following sets:

$$M^i = \{x\in\mathbb{F}_2^n\ |\ v \cdot \Phi_F(x) = 0,\  x\cdot e^i = 0\};$$
$$M^j = \{x\in\mathbb{F}_2^n\ |\ v \cdot \Phi_F(x) = 0,\  x\cdot e^j = 0\};$$
$$M^{ij} = \{x\in\mathbb{F}_2^n\ |\ v \cdot \Phi_F(x) = 0,\  x\cdot (e^i + e^j) = 0\}.$$

Then, we have
$$\sum_{b\preceq(a^v+ {\bf 1})} W_{v \cdot \Phi_F}(b) = 4 |M^i| - 2^n + 4 |M^j| - 2^n + 4 |M^{ij}| - 2^n = 4 (|M^i| + |M^j| + |M^{ij}|) - 3\cdot2^n$$
$$= 4 ( 2^{n-1} + 2|M^{ij}_{0}|) - 3\cdot2^n = 8|M^{ij}_{0}| - 2^{n-1},$$
where $$M^{ij}_0 = \{x\in\mathbb{F}_2^n\ |\ v \cdot \Phi_F(x) = 0,\  x\cdot e^i = 0,\ x\cdot e^j = 0\}.$$

{\bf Step 2.} Thus, we need to prove that there exist coordinates $i, j$, $1\leqslant i\neq j\leqslant n$, such that $|M^{ij}_{0}|$ is odd (since we consider $n\geq 5$).
From \cite{19-gorodilova} (prop. 7), we know that $M = \{x\in\mathbb{F}_2^n\ |\ v\cdot \Phi_F(x) = 0\} = \bigcup_{{\ell}\in I} A_{\ell}$, where $A_{\ell}$ is a linear subspace of dimension 2, and $A_{\ell}\cap A_k = \{{\bf 0}\}$, ${\ell},k\in I$, ${\ell}\neq k$. Since $\Phi_F$ is a permutation, then $|M|=2^{n-1}$ and $|I| = (2^{n-1}-1)/3$.

Let us consider an arbitrary $A_{\ell} = \{{\bf 0}, x^{\ell}, y^{\ell}, x^{\ell} + y^{\ell}\}$.
Then for any distinct coordinates $i,j$ of $x^{\ell}, y^{\ell}, x^{\ell} + y^{\ell}$ we have the following situations (without permutations of rows):
\begin{center}
\begin{tabular}{cccccccccc}
&                 $ij$ &    & $ij$ & & $ij$  & & $ij$  & & $ij$  \\
$x^{\ell}$        &   00 &    & 00 &    & 00 &    & 00 &    & 01 \\
$y^{\ell}$        &   00 & or & 01 & or & 10 & or & 11 & or & 10 \\
$x^{\ell} + y^{\ell}$ & 00 &    & 01 &    & 10 &    & 11 &    & 11 \\
\end{tabular}
\end{center}
Hence, the number of $x^{\ell}, y^{\ell}, x^{\ell} + y^{\ell}$ together with ${\bf 0}$ that belong to the set $M^{ij}_0$ is equal to  $1 + 3\cdot N^{ij}_3 + 1\cdot N^{ij}_1 + 0\cdot N^{ij}_0$, where $N^{ij}_3 + N^{ij}_1 + N^{ij}_0 = |I| = (2^{n-1}-1)/3$, and $N^{ij}_k$, $k=0,1,3$, is the number of $A_{\ell}$, $\ell\in I$, having exactly $k$ vectors with both coordinates $i$ and $j$ equal to~$0$.

Thus, $|M^{ij}_0|$ is odd if and only if $N^{ij}_0$ is odd.

{\bf Step 3.} Now, we need to prove that there exist coordinates $i, j$, $1\leqslant i\neq j\leqslant n$, such that $N^{ij}_0$ is odd. We found the following interesting property (computationally verified for $n=5$) that we formulate as a conjecture.

\begin{conjecture}\label{sets}
Let $M = \bigcup_{{\ell}\in I} A_{\ell}$, where $A_{\ell}$ is a linear subspace of dimension 2, and $A_{\ell}\cap A_k = \{{\bf 0}\}$, ${\ell},k\in I$, ${\ell}\neq k$, $|I| = (2^{n-1}-1)/3$. Then the set $M$ is a hyperplane $\{x\in\mathbb{F}_2^n\ |\ x_m = 0\}$ for some coordinate $m$ if and only if the number of subspaces $A_{\ell}$ without elements having both coordinates $i$ and $j$ equal to $0$ is even for any distinct coordinates $i,j$.
\end{conjecture}

{\bf Step 4.} If conjecture~\ref{sets} is true, then we need to prove that $M = \{x\in\mathbb{F}_2^n\ |\ v\cdot \Phi_F(x) = 0\}$ cannot be a hyperplane $\{x\in\mathbb{F}_2^n\ |\ x_m = 0\}$ for some coordinate $m$.

\begin{conjecture}\label{subspace}
Let $F$ be a quadratic APN function in $n$ variables, $n\geq5$. Then $\{x\in\mathbb{F}_2^n\ |\ v\cdot \Phi_F(x) = 0\}$ is not a linear subspace.
\end{conjecture}

We computationally verified this property for all known quadratic APN functions for $n=5,\ldots,11$ and formulate the conjecture.

Thus, by proving conjectures~\ref{sets} and \ref{subspace}, we can prove the starting conjecture~\ref{deg}. Unfortunately, each of them remains open up to now.

\bigskip
\section*{Conclusion}

The following question is open: what properties must a Boolean function satisfy in order to be the associated function for some vectorial function? Even a partial answer to the question provides a potential method to find new APN functions if we can choose ``admissible'' Boolean functions as $\gamma_F$. For example, using the algorithm from \cite{19-boura} for reconstructing APN functions from its associated functions. Another reason why we study the properties of associated functions is that they may lead to new results in the differential equivalence classification of APN functions.

\end{document}